\documentclass[journal,onecolumn]{IEEEtran}

\def\P{\mathbb{P}}

\newcommand{\argmin}{\operatornamewithlimits{argmin}}

\def\p{{\mathbb{P}}}
\def\E{{\mathbb{E}}}

\usepackage{amsmath}
\usepackage{amsthm}
\usepackage{amssymb}
\usepackage{colordvi}
\usepackage{amsfonts}
\usepackage{graphicx}
\usepackage{epsfig}
\usepackage{color}
\usepackage{url}
\usepackage{wrapfig}
\usepackage{tikz}
\usetikzlibrary{arrows,backgrounds,snakes}
\usepackage[caption=false]{subfig}

\makeatletter
\renewcommand*{\p@section}{\S\,}
\renewcommand*{\p@subsection}{\S\,}
\makeatother

\usetikzlibrary{arrows,matrix,backgrounds,fit,calc,automata,shapes,positioning,decorations.pathreplacing}
\usetikzlibrary{calc}

\tikzset{
    between/.style args={#1 and #2}{
         at = ($(#1)!0.5!(#2)$)
    }
}

\newtheorem{theorem}{Theorem}

\newtheorem{definition}[theorem]{Definition}

\begin{document}

\title{Contrasting Effects of Replication in Parallel Systems: From Overload to Underload and Back}

\author{Felix~Poloczek,~\IEEEmembership{University of Warwick / TU Berlin}
        and~Florin~Ciucu,~\IEEEmembership{University of Warwick}}

\maketitle
\begin{abstract}
Task replication has recently been advocated as a practical
solution to reduce latencies in parallel systems. In addition to
several convincing empirical studies, some others provide
analytical results, yet under some strong assumptions such as
Poisson arrivals, exponential service times, or independent
service times of the replicas themselves, which may lend
themselves to some contrasting and perhaps contriving behavior.
For instance, under the second assumption, an overloaded system
can be stabilized by a replication factor, but can be sent back in
overload through further replication. In turn, under the third
assumption, strictly larger stability regions of replication
systems do not necessarily imply smaller delays.

Motivated by the need to dispense with such common and restricting
assumptions, which may additionally cause unexpected behavior, we
develop a unified and general theoretical framework to compute
tight bounds on the distribution of response times in general
replication systems. These results immediately lend themselves to
the optimal number of replicas minimizing response time quantiles,
depending on the parameters of the system (e.g., the degree of
correlation amongst replicas). As a concrete application of our
framework, we design a novel replication policy which can improve
the stability region of classical fork-join queueing systems by
$\mathcal{O}(\ln K)$, in the number of servers $K$.
\end{abstract}

\section{Introduction}
Despite a significant increase in network bandwidth and computing
resources, major online service providers (and not only) still
face extremely volatile revenues due to the high variability of
latencies (aka response times/delays), especially in their tails
(e.g., the $95^\textrm{th}$-percentile). Several well-cited and
convincing studies reported significant potential revenue loss by
Google, Bing, or Amazon, were the latencies
higher~\cite{Schurman09,Hoff09,Souders09}; a typical cited
argument is that an additional $100$ms in latency would cost
Amazon $1\%$ of sales.

Given the late abundance of computing resources, a natural and yet
very simple way to improve latencies is \textit{replication}, a
concept which was traditionally used to improve the reliability of
fault-tolerant systems~\cite{Schneider90}. In the context of a
multi-server (parallel) system, the idea is merely to replicate a
task into multiple copies/replicas, and to execute each replica on
a different server. By leveraging the statistical variability of
the servers themselves, as execution platforms, it is expected
that some replicas would finish much faster than others; for a
discussion of various system/OS factors affecting execution times
see~\cite{Dean13}. The key gain of executing multiple replicas is
not to reduce the average latency, but rather the latency tail
which is recognized as critically important for ensuring a
consistently fluid/natural responsiveness of systems. Therefore,
replication can be regarded as being instrumental to the
development of ``latency tail-tolerant systems", similarly to its
role in fault-tolerant systems~\cite{Dean13}.

While the idea of using redundant requests is not new, as it has
been used to demonstrate significant speedups in parallel programs
\cite{Ghare04,Heymann01}, it has become very attractive with its
implementation in the MapReduce framework through the so-called
`backup-tasks'~\cite{Dean08}. Thereafter there has been a surge of
very high-quality empirical work which has convincingly
demonstrated the benefits of using redundancy for significant
latency improvement, both in the mean and also top percentiles.
Such works include latency reductions in Google's distributed
systems~\cite{Dean12}, in DNS queries and database
servers~\cite{Vulimiri13}, key-value storage
systems~\cite{Stewart13}, cloud storage systems~\cite{Wu15}, or
significant speed-ups of small jobs in
data-centers~\cite{Ananthanarayanan13} or short TCP
flows~\cite{XuLi14}.

Such empirical work has been complemented by several excellent
analytical studies (see the Related Work section), which have
provided fundamental insight into the benefits of replication.
Constrained by analytical tractability, most of these works make
several strong assumptions: not only the arrivals are Poisson and
the service times are exponentially distributed (i.e., typical
assumptions in the queueing literature), but the service times of
the replicas plus the corresponding original tasks are
statistically independent. By challenging these assumptions,
especially the last two, we first provide some elementary
analytical arguments, along with some simulation results, that the
benefits of replication are highly dependent on both the
distributional and correlation structures of the service times. A
convincing example is that the stability region of a system is not
monotonous in the replication factor. For instance, by adding a
replica server an overloaded system can be stabilized, an
advantage which however vanishes by adding additional replica
servers, whence the subtitle of this paper.

The contribution of this paper is a general analytical framework
to compute stochastic bounds on the response time distributions in
replication systems. In particular, our framework covers scenarios
with Markovian arrivals, general service time distributions
(subject to a finite moment generating function), and a
correlation model amongst the original and replicated tasks. Using
back-of-the-envelope calculations, our results can be immediately
used for engineering purposes (e.g., to determine the optimum
number of replicated servers to minimize the top percentiles of
latencies). A key feature of our methodology is that it relies on
a powerful martingale methodology which was recently shown to
provide remarkably accurate stochastic bounds in various and
challenging queueing systems with non-Poisson arrivals~(see,
e.g.,~\cite{CiPoSc13Inf,PoCi14,RiPoCi15}). According to several
numerical/simulation illustrations, our results exhibit a similar
high accuracy, including the challenging case of Markovian
arrivals.

To concretely illustrate the applicability of our results we
consider two applications. The first is to improve the performance
of MapReduce systems (which can be abstracted by a Fork-Join (FJ)
queueing model) through replication. In particular, we design an
elementary replication policy which can significantly improve not
only delay quantiles (e.g., by a factor of roughly $2$), but more
fundamentally the stability region of a FJ system by a logarithmic
factor $\mathcal{O}(\ln K)$ in the number of servers $K$; our
analysis provides a theoretical understanding of the benefits of
using back-up tasks in MapReduce, as a proposal to alleviate the
problem of stragglers~\cite{Dean08}. Albeit such a theoretical
benefit is obtained under strong exponential and statistical
independence assumptions, simulation results show that the
underlying numerical benefits carry over to realistic scenarios
subject to correlations amongst replicas. The second application
investigates the analytical tradeoff between resource usage and
response times under replication, a matter which has recently been
addressed through Google and Bing empirical studies. The key
analytical insight is that increasing resource usage through
replication yields a substantial reduction of response time upper
quantiles if the service times of the replicas are sufficiently
independent (i.e., subject to a low correlation factor, to be
later formally described).

The rest of the paper is organized as follows. In
Section~\ref{sec:rmrw} we introduce the analytical models and
discuss related work. In Section~\ref{sec:eai} we provide several
insights into the benefits of replication, by following elementary
models and derivations. In Section~\ref{sec:theory} we provide our
general theoretical framework dealing with both Poisson and
Markovian arrivals, and also independent and correlated replicas
(i.e., four scenarios). In Section~\ref{sec:app} we investigate
the two applications of our analytical framework. In
Section~\ref{sec:con} we conclude the paper.

\section{Replication Models and Related Work}\label{sec:rmrw}
We consider a parallel system with $K$ homogeneous servers with
identical speeds (see Figure~\ref{fig:replica}). A stream of tasks
arrives at a dispatcher according to some stationary and ergodic
point process; the interarrival times are denoted by $t_i$ with
the mean $E\left[t_1\right]=\frac{1}{\lambda}$, whereas their
number within the (continuous) time interval $(0,t]$ is denoted by
$N(t)$. This process can have a Markov structure, to be more
precisely defined in Section \ref{subsec:mair}.

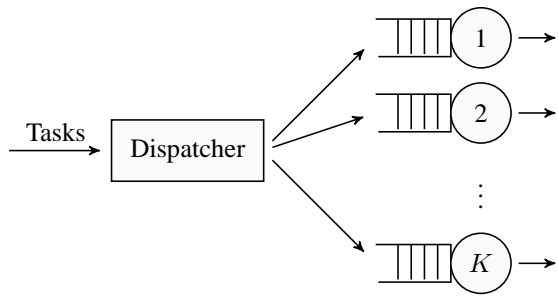
\begin{figure}[t]
\begin{center}
\tikzstyle{every state}=[circle,fill=black!2,minimum
size=0.8cm,inner sep=0pt]
\begin{tikzpicture}[->,>=stealth',shorten >=1pt,auto,node distance=1.3cm,semithick]
\node[draw=none] at (6.4,-2.25cm) (9) {$\vdots$};
\node[draw,rectangle,fill=black!2,inner sep=0.25cm] at
(2.5,-1.75cm) (11) {Dispatcher}; \node at (3.5,-1.75cm) (12) {};
\node at (1.5,-1.75cm) (5) {}; \node at (0,-1.75cm) (6) {};
\draw[->] (6)  edge node {Tasks}   (5); \foreach \j in {1,2,4} {
  \draw[-] (5,-\j+1) -- ++(1cm,0) -- ++(0,-0.5cm) -- ++(-1cm,0);
  \foreach \i in {1,...,4}
    \draw[-] (6cm-\i*5pt,-\j+1) -- +(0,-0.5cm);
  \node[state] at (6.4,-\j-0.25+1) (\j+6) {\if \j4 {$K$} \else {\j}\fi};
  \draw[->] (6.9,-\j-0.25+1) -- +(15pt,0);
  \node at (5,-\j-0.25+1) (\j) {};
  \draw[->] (12)  edge node {}   (\j);
}
\end{tikzpicture}
\caption{A parallel system with $K$ servers; tasks are dispatched
to the servers in a possibly replicated manner (i.e., the same
task to multiple servers)} \label{fig:replica}
\end{center}
\end{figure}

The service times of the tasks are denoted by $x_i$ and are drawn
from some general distribution subject to a finite moment
generating function; the average is set to
$E\left[x_1\right]=\frac{1}{\mu}$. For numerical purposes, we will
occasionally use the analytically convenient Pareto distribution,
which can be approximated within our theoretical framework through
a hyperexponential distribution.

The \textit{utilization} of one server, in a system without
replicas where tasks are symmetrically distributed, is denoted by
\begin{equation*}
\rho:=\frac{\lambda}{K\mu}~.
\end{equation*}
In general, it is assumed for stability that $\rho<1$. However, in
a system with replication, the expression of the utilization
$\rho$ may change depending on various factors (e.g., the
distribution of tasks' service times) whereas the stability
condition may fail (such occurrences will be specifically
indicated).

\subsection{Tasks Assignment Policies}\label{sec:tap}
A crucial design component in the parallel server system is the
task assignment policy, i.e., how are the incoming tasks assigned
to the $K$ servers for processing? While many such policies have
been analytically and empirically studied, we focus on few
relevant ones in terms of both performance and overhead:
\begin{itemize}
\item{\textbf{Random}}:~Each task is dispatched, uniformly at
random, to one of the $K$ servers; in the particular case of a
Poisson (overall) arrival stream, the tasks arrived at some server
follow a Poisson distribution with rate $\frac{\lambda}{K}$.

\item{\textbf{Round-Robin}}:~Tasks are deterministically
dispatched in a circular fashion to the $K$ servers, i.e., task
$i$ is assigned to server $i~\textrm{mod}~K$ (with the convention
that $0$ stands for $K$); in the case of a Poisson stream, the
interarrival times at some server follow an Erlang $E(K,\lambda)$
distribution.

\item{\textbf{G/G/K}}:~Unlike the previous two schemes, which
immediately dispatch the incoming tasks, and whereby tasks enqueue
at the assigned servers, in $G/G/K$ it is the responsibility of
each server to fetch a single task, from a centralized queue at
the dispatcher, once they become idle.

\item{\textbf{(Full-)Replication}} ($K$-replication factor):~Each
incoming task $i$ is replicated to all the $K$
servers\footnote{For the sake of clarification, the original task
is called a replica as well.}; the corresponding service times are
denoted by $x_{i,j}$ for $j=1,\dots,K$. Alike in \textit{Random}
and \textit{Round-Robin}, each server maintains a local (FIFO)
queue.

\item{\textbf{Partial-Replication}} ($k$-replication
factor):~Besides \textit{full} replication, a task may be
replicated to only $k\leq K$ servers; for simplicity, we will
assume that both $K$ and $k$ are powers of $2$, and that
consecutive blocks of $k$ replicas are allocated to the $K$
servers in a round-robin manner. We call the underlying strategy
(strict) \textit{Partial-Replication} when $1<k<K$, and
\textit{No-Replication} when $k=1$.
\end{itemize}

In terms of analytical tractability, \textit{Random} and
\textit{Round-Robin} are significantly more amenable than $G/G/K$;
in fact, exact results are known for $G/G/K$ only in the case of
Poisson arrivals and exponential service times (in which case the
model is denoted by $M/M/K$). However, $G/G/K$ yields
significantly better performance (i.e., much smaller response
times of the tasks) than \textit{Random} and \textit{Round-Robin},
especially in the case of high variability of the tasks' service
times; in turn \textit{Round-Robin} slightly outperforms
\textit{Random} (for an excellent related discussion
see~\cite{MorHarchol12}, pp. 408-430).

It is to be noted however that the superiority of $G/G/K$ is
(partly) due to the availability of additional system information,
i.e., each task is `informed' about which server is idle such that
it can minimize its response time. In turn, amongst policies which
are oblivious to such information, \textit{Round-Robin} was shown
to be optimal for exponential~\cite{Ephremides80,Walrand88} and
increasing failure rate distributions~\cite{Liu94}; for a recent
state-of-the-art queueing analysis of Round-Robin
see~\cite{Hyytia13}.

A more sophisticated replication strategy was proposed in the
context of massively parallel data processing systems in which
(large) jobs are forked/split into (smaller) tasks, each assigned
to a server; once a fraction of the tasks finish their executions,
each of the remaining (and straggling) tasks are further
replicated. This model appeared in the MapReduce
specification~\cite{Dean08}, and was formally studied in terms of
the underlying response time / resource usage tradeoff, albeit by
disregarding queueing effects in~\cite{WangJW15}. Another strategy
used by Google is to defer the start of executing the second
replica for some suitable time, in order to reduce resource
usage~\cite{Dean13}.

\subsection{Purging/Cancellation Models}
Before discussing the relative performance of \textit{Replication}
to other policies, we first define how replication strategies deal
with residual resources.
\begin{itemize}
\item{\textbf{Purging}}:~A task is considered to complete (and
hence its response time is determined) when the fastest replica
finishes its execution; at the same time, the residual replicas
are all purged/cancelled from the system (with some negligible
related cost).

\item{\textbf{Non-Purging}}:~A task response time is determined as
in the \textit{Purging} case, but the remaining replicas leave the
system no sooner than their execution end.
\end{itemize}

\textit{Purging} is clearly more efficient from a purely
\textit{task response-time} perspective, as it frees resources
once the first replica completes; this operation demands however
synchronization overhead amongst the servers. One basic reason for
this superiority is that in the \textit{Non-Purging} model the
utilization increases $k$-fold for a $k$-replication factor, for
any task service time distribution; in particular, a
$2$-replication factor requires the replica-free system to have a
utilization under $50\%$ (otherwise the response times get
unbounded). In turn, the growth of the utilization is less
pronounced in the \textit{Purging} model, depending on the type of
distribution of the service times; in fact, and perhaps
counterintuitively, there is no increase in the case of the
exponential distribution regardless the replication factor (for a
follow-up discussion see \ref{sec:stability}).

Besides the advantage of a better queueing performance, the
\textit{Purging} model is much easier to analyze. In fact, the
only analytical study of \textit{Non-Purging} is considered
in~\cite{Vulimiri13}; besides the classical and simplifying
assumptions of Poisson arrivals and exponential service times, the
underlying queueing analysis critically relies on an artificial
statistical independence assumption amongst the queues. Using this
assumption, it is shown that below a utilization threshold of
$33\%$, a $2$-replication factor strategy does improve the
response time despite the inherent doubling of the utilization.

A generalized version of \textit{Partial-Replication} considers
the situation when the fastest $l\leq k$ replicas finish their
execution (the residual ones being subsequently purged); a
practical use of this generalization is in coded distributed
storage systems~\cite{Shah13}. The central result is that under
arrivals with \textit{independent increments}, and exponential (or
`\textit{heavier}') service times, \textit{Full-Replication}
minimizes the (average) response times. In turn, in the case of
`\textit{lighter}' service times and $100\%$ utilization, a
replication factor greater than one is detrimental. The underlying
proofs use an ingenious coupling argument, but do not provide
quantitative results.

Another set of qualitative results, on the superiority of
\textit{Full-Replication} for a specific type of service time
distributions (including the exponential) is presented
in~\cite{Koole2007}. Interestingly, under a discrete time model
with geometric service time distributions, is is shown
in~\cite{Borst03} through quantitative results that
\textit{No-Replication} is optimal (for an explanation of the
apparent contradiction between exponential and geometric service
time distributions, with respect to the optimality of the
replication model, see~\cite{Koole2007}).

Recently, an \textit{Early Purging} model, in which residual
replicas are purged once the first one starts its execution, has
been mentioned in~\cite{Dean13} and further analyzed
in~\cite{Joshi15}; besides reducing the resource usage, it was
shown that this model can also significantly reduce response times
despite the apparent loss of diversity, at high utilizations.

The perhaps most fundamental related result obtained so far is a
recent exact analysis under the purging model~\cite{Gardner15}.
While the analysis critically relies on the Poisson/exponential
models, a key analytical contribution is capturing multi-class
arrivals (i.e., different arrival streams are served by different
sets of (replicated) servers). The elegance of the results lends
itself to several fundamental and contriving insights into the
properties of replication, especially accounting for the
multi-class feature of the model.

More general stochastic bounds in replication systems are obtained
in~\cite{FiJi15}, including the very challenging multi-stage case,
by leveraging the analytical power of the stochastic network
calculus methodology. While the underlying arrival and service
models from~\cite{FiJi15} are more general than ours, the crucial
difference is in handling the underlying correlation structures:
concretely, while~\cite{FiJi15} deals with arbitrary correlation
structures yielding stochastic bounds holding in great generality,
we exploit the specific correlation structures through the
martingale methodology.

\section{Elementary analytical Insights}\label{sec:eai}
Here we complement the previous discussion by providing several
motivating examples. After quickly contrasting the task assignment
policies introduced earlier, under the Poisson/exponential models,
we explore more general service time distributions. The key
insight is that the stability region of replicated systems is not
necessarily monotonous in the number of replicas; depending on the
service distribution, any of the policies \textit{No-Replication},
\textit{Full-Replication}, or \textit{Partial-Replication} can
yield the largest stability region.

\subsection{The M/M model} For some immediate analytical insight, consider the classical example of
Poisson arrivals and exponential service times. Due to a lack of
closed-form formulas for all considered policies, for large number
of servers, we assume that $K=2$; recall that the (server)
utilization is $\rho=\frac{\lambda}{2\mu}$.

The average response times for the four policies (i.e.,
\textit{Random}, \textit{Round-Robin}, \textit{M/M/2}, and
\textit{Replication}) are, respectively,
\begin{eqnarray*}
E\left[T_{Rnd}\right]&=&\frac{1}{\mu(1-\rho)}\\
E\left[T_{RR}\right]&=&\frac{2}{\mu\left(1-4\rho+\sqrt{1+8\rho}\right)}\\
E\left[T_{MM2}\right]&=&\frac{1}{\mu\left(1-\rho^2\right)}\\
E\left[T_{Rep}\right]&=&\frac{1}{2\mu(1-\rho)}~.
\end{eqnarray*}
Note that \textit{Replication} induces an $M/G/1$ queueing model,
in which the service time is the first order statistics of two
i.i.d. random variables (in the current case being an exponential
with half of the mean of the original). Immediate comparisons
reveal that the minimum (`best') response time is attained by
\textit{Replication}; a key reason is that the gain of sampling
the minimum of exponential random variables, together with the
\textit{Purging} model, significantly dominates the cost of
temporary redundant resource usage. In turn, the maximum (`worst')
response time is attained by \textit{Random}; the relative
performance of \textit{Round-Robin} and \textit{M/M/2} depends on
the value of $\rho$. Lastly, we point out that the superiority of
\textit{Replication} immediately extends to larger values of $K$.

More general results in terms of lower and upper bounds on the
average response time in the case of a variant of
\textit{Replication}, in which only the fastest $l\leq K$ tasks
are required to complete (whilst the residual tasks are purged)
(and which was \textit{qualitatively} studied in~\cite{Shah13}),
appeared in~\cite{Joshi14}; in particular, it was shown that
\textit{Replication} outperforms the corresponding \textit{M/M/K}
model. Further upper bounds were derived in the case of general
service time distributions, using existing bounds on the first two
moments of the $l^{\textrm{th}}$ order statistics.

\subsection{Beyond the M model}\label{sec:stability}
In the previous example with exponential service times, the
stability region is invariant to the replication factor; the
reason is that the $1^{\textrm{st}}$ order statistic of $K$
(independent) exponential random variables $exp(\mu)$ is an
exponential random variable $exp(K\mu)$. The next elementary
examples show that any strategy amongst \textit{No-Replication},
\textit{Full-Replication}, or \textit{Partial-Replication} can
yield the strictly largest stability regions (and hence `best'
response times, at least in some subset of the stability region; a
follow-up discussion will be given in Section~\ref{subsec:iacr}).
A fundamental reason is the assumption of independent service
times of the replicas, which motivates the need for accounting for
some correlation structures.

Recall that in the \emph{No-Replication} scenario, a necessary and
sufficient condition for stability (or, equivalently, for finite
response times) is
\begin{equation*}
 \E[x_1]<K\E[t_1]~.
\end{equation*}

In the case of \emph{Full-Replication}, the corresponding
stability condition is given by
\begin{equation*}
 \E\left[\min\left\{x_1,\ldots, x_n\right\}\right]<\E[t_1]~,
\end{equation*}
whereas in the case of \emph{Partial-Replication} with
replication factor $k$ by
\begin{equation}\label{eq:stability}
 \E\left[\min\left\{x_1,\ldots, x_k\right\}\right]<\frac{K}{k}\E[t_1]~.
\end{equation}

Denoting the CCDF of $x_i$ by
\begin{equation*}
 f(x):=\P(x_1\ge x)~,
\end{equation*}
we observe from the previous stability conditions that the `best'
replication-factor $k$ is
\begin{equation}
\argmin_k k\int f^k(x)dx~.\label{eq:bestk}
\end{equation}

We next present examples of different distributions for $x_i$
resulting in `best' scenarios for each of the three replication
strategies.

\subsubsection{No-Replication: Uniform} Assume uniformly
distributed service times, i.e., $x_i\sim\mathcal{U}_{[0,1]}$. The
following argument shows that in this case replication is
detrimental, i.e.,
\begin{equation*}
 \E\left[x_1\right]<k\E\left[\min\left\{x_1,\ldots,x_k\right\}\right]~,
\end{equation*}
for any $k\ge 2$~:
\begin{align*}
 k\E\left[\min \left\{x_1,\ldots,x_k\right\}\right] = & k\int_0^{\infty}\P\left(\min \left\{x_1,\ldots,x_k\right\}\ge x\right)dx\\
 = & k\int_0^{\infty}\P\left(x_1\ge x\right)^k dx\\
 = & \int_0^{1}kx^kdx = \frac{k}{k+1}>\frac{1}{2}=\E\left[x_1\right]~.
\end{align*}
The same argument additionally shows that
\textit{Partial-Replication} is better than
\textit{Full-Replication}. This result extends the qualitative
observation from~\cite{Shah13} (i.e., Theorem~4 therein,
restricted to a $100\%$ utilization, and hence an unstable regime)
to any (stable) utilization.

\subsubsection{Full-Replication: Weibull} Let the $x_i$ now be
Weibull distributed, i.e.,
$f(x)=e^{-\left(x/\lambda\right)^\alpha}$. For $\alpha<1$, a
higher degree  of replication is `better', as shown below:
\begin{align*}
 k\E\left[\min\left\{x_1,\ldots,x_k\right\}\right] = & k\int_0^{\infty}\P(\min \left\{x_1,\ldots,x_k\right\}\ge x)dx\\
 = & k\int_0^{\infty}e^{-k\left(x/\lambda\right)^{\alpha}}dx\\
 = & k\frac{\lambda}{k^{1/\alpha}}\Gamma(1+1/\alpha)~.
\end{align*}
By the assumption on $\alpha$, the last term is monotonically
decreasing in $k$~. Note that in the special case of exponentially
distributed $x_i$, i.e., $\alpha=1$, replication is neither
beneficial nor detrimental (from the point of view of the
stability region), as pointed out earlier. This result also
extends the qualitative observation from~\cite{Shah13} (i.e.,
Theorem~3) to any (stable) utilization.

\subsubsection{Partial Replication: Pareto} Lastly we consider the
Pareto distribution, i.e., $f(x)=x^{-\alpha}$ for $x\ge 1$. For a
suitably chosen $\alpha>1$, it can be shown that (strict)
\textit{Partial-Replication} can become `better' than both
\textit{Full-Replication} and \textit{No-Replication}:
\begin{align*}
 k\E\left[\min \left\{x_1,\ldots,x_k\right\}\right] = & k\int_0^{\infty}\P\left(\min \left\{x_1,\ldots,x_k\right\}\ge x\right)dx\\
 = & k+k\int_1^{\infty} x^{-k\alpha}dx=k+\frac{k}{k\alpha-1}~.
\end{align*}
It is clear that for sufficiently small $\alpha>1$, the minimal
value is attained for $k=2$~.

This last example highlights that the performance of replication
strategies heavily depends on the replication factor $k$, the service
time distribution, and other underlying assumptions. In
particular, performance is not monotonic in $k$, and thus an
optimization framework is desirable (related results, on the
actual response time distributions as a function of $k$ will be
provided in the next section).

\begin{figure}[t]
\begin{center}
\includegraphics[width=0.48\textwidth]{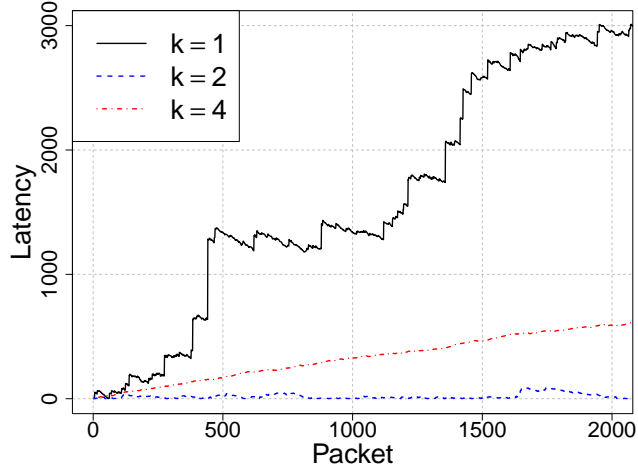}
\caption{From overload ($k=1$) to underload ($k=2$) and back
($k=4$) ($K=4$, $\alpha=1.1$, $\lambda=1$, and utilization
$\rho=2.75$ (for the non-replicated $k=1$
case))}\label{fig:pareto}
\end{center}
\end{figure}

For complementary numerical results illustrating the
counterintuitive effect of $k$, consider the Pareto distribution
with the assumption of independent service times of the $k$
replicas. Let $K=4$, arrival rate $\lambda=1$, $\alpha=1.1$ (for
the Pareto distribution), yielding a utilization $\rho=2.75$
(i.e., $275\%$). By plotting the simulated latencies of the first
$10^4$ packets, Figure~\ref{fig:pareto} shows that while the
system without replication is in overload, a replication factor of
$k=2$ stabilizes the system (reducing the utilization to $0.91$),
whereas a replication factor of $4$ puts the system back in
overload (increasing the utilization to $1.29$).

The non-monotonic behavior in $k$ disappears when the service
times are sufficiently correlated. Indeed, by taking the service
times of the replicas as $y+x_i$ (where the $x_i$ are Pareto distributed, and $y\ge 0$ is arbitrary), it holds:
\begin{align*}
 k\E[\min\{y+x_1,\ldots,y+x_k\}] & = k\E[y] + k\E[\min\{x_1,\ldots,x_k\}]\\
 & = k\E[y] + k+\frac{k}{k\alpha-1}\\
 & = k\left(\E[y] + \frac{k\alpha}{k\alpha-1}\right)~,
\end{align*}
so that (for a suitably chosen $\alpha>1$, and a sufficiently large value of $\E[y]$) the optimal value of $k$ in Eq.~(\ref{eq:bestk}) is $1$ (i.e., \textit{No-Replication} is `best').

\section{Theory}\label{sec:theory}
We assume a queueing system with $K$ servers and interarrival
times between jobs $i$ and $i+1$ denoted by $t_i$. Upon its
arrival, job $i$ is replicated to $k\le K$ servers where they are
processed with service times $x_{i,1},\ldots,x_{i,k}$,
respectively. For simplicity, we throughout assume that $K$ is an
integral multiple of $k$. Further, the jobs are assigned to the
$\frac{K}{k}$ batches in a round robin scheme, i.e. the
interarrival times for one batch can be described as:
\begin{equation*}
 \tilde{t}_i:=\sum_{j=0}^{K/k-1} t_{\left(i-1\right)\frac{K}{k}+j}~.
\end{equation*}

The following recursion describes the response time $r_{i+1}$ of
job $i+1$, i.e., the time between the job's arrival and its
service being complete:
\begin{equation*}
 r_1:=\min_{j\le k}x_{1,j}~,\quad r_{i+1}:=\min_{j\le k}\{x_{i+1,j}\}+\max\{0,r_{i}-\tilde{t}_{i}\}~,
\end{equation*}
resulting in a representation of the \emph{steady-state} response
time $r$ as:
\begin{equation}\label{eq:responsetime}
 r=_{\mathcal{D}}\max_{n\ge 1}\left\{\sum_{i=1}^{n+1} \min_{j\le k}\{x_{i,j}\}-\sum_{i=1}^n \tilde{t}_i\right\}~,
\end{equation}
where $=_{\mathcal{D}}$ stands for equality in distribution, and
the empty sum is by convention equal to $0$.

Depending on the correlation between either the interarrival times
and the service times, respectively, we consider four different
scenarios: In Subsection~\ref{subsec:iair}, all random variables
$t_i$, $x_{i,j}$ are assumed to be independent. In
Subsection~\ref{subsec:mair}, the interarrival times are driven by
a certain Markov chain, whereas in Subsection~\ref{subsec:iacr}
the service times are correlated through a common additive factor.
Finally, in Subsection~\ref{subsec:macr}, a combination of both
correlation models is considered.

\subsection{Independent Arrivals, Independent Replication}\label{subsec:iair}

As stated above, we consider the scenario of \emph{independent
replication}, i.e., $\left\{t_i,x_{i,j}\;\middle|\;i\ge 1,\;j\le
k\right\}$ is an independent family of random variables.

The next Theorem provides an upper bound on the CCDF of $r$ as
defined in Eq~(\ref{eq:responsetime}):
\begin{theorem}
\label{tm:iair}
 Let $\theta_{\text{ind}}$ be defined by
 \begin{equation*}
  \theta_{\text{ind}}:=\sup\left\{\theta\ge 0\;\middle|\;\E\left[e^{\theta\min_{j\le k}\{x_{i,j}\}}\right]\E\left[e^{-\theta t_i}\right]^{\frac{K}{k}}\le
  1\right\}~.
 \end{equation*}
 Then the following bound on the response time holds for all $\sigma\ge 0$:
 \begin{equation*}
  \P(r\ge\sigma)\le \E\left[e^{\theta_{\text{ind}}\min_{j\le k}\{x_{1,j}\}}\right]e^{-\theta_{\text{ind}}\sigma}~.
 \end{equation*}
\end{theorem}

Note that, given the stability conditions from
Eq.~(\ref{eq:stability}), $\theta_{\text{ind}}>0$ as
\begin{align*}
 \frac{d}{d\theta} & \left.\E\left[e^{\theta\min_{j\le k}\{x_{i,j}\}}\right]\E\left[e^{-\theta t_i}\right]^{\frac{K}{k}}\right|_{\theta=0}\\
 & = \E\left[\min_{j\le k}\{x_{i,j}\}\right]-\frac{K}{k}\E\left[t_i\right] < 0~.
\end{align*}

\begin{proof}
 Define the process $M(n)$ by
 \begin{equation*}
  M(n+1):=e^{\theta_{\text{ind}}\left(\sum_{i=1}^{n+1} \min_{j\le k}\{x_{i,j}\}-\sum_{i=1}^n \tilde{t}_i\right)}~.
 \end{equation*}
 $M(n)$ is a martingale:
 \begin{align*}
  \E & [M(n+1)M(n)^{-1}\mid M(1),\ldots,M(n)]\\
  & = \E\left[e^{\theta_{\text{ind}}\left(\min_{j\le k}\{x_{i,j}\}-\tilde{t}_n\right)}\right]\\
  & = \E\left[e^{\theta_{\text{ind}}\min_{j\le k}\{x_{i,j}\}}\right] \E\left[e^{-\theta_{\text{ind}}t_n}\right]^{\frac{K}{k}}\\
  & = 1~.
 \end{align*}
 Now define the stopping $N$ as
 \begin{equation*}
  N:=\min\left\{n\ge 0\;\middle|\; \sum_{i=1}^n\min_{j\le k}\{x_{i,j}\}-\sum_{i=1}^{n-1} \tilde{t}_i\ge\sigma\right\}~,
 \end{equation*}
 and note that $\{N<\infty\}=\{r\ge\sigma\}$. With the optional stopping theorem
 \begin{align*}
  \E\left[e^{\theta_{\text{ind}}\min_{j\le k}\{x_{i,j}\}}\right ] & =\E\left[M(1)\right]\\
  & =\E\left[M(N\wedge l\right]\\
  & \ge \E\left[M(N\wedge l)1_{N\le l}\right]\\
  & \ge e^{\theta_{\text{ind}}\sigma}\P(N\le l)
 \end{align*}
 Now let $l\to\infty$.
\end{proof}

We point out that the proof essentially follows the bounding
technique for GI/GI/1 queues from~\cite{Kingman64}.

\subsection{Markovian Arrivals, Independent Replication}\label{subsec:mair}

We now turn to the more realistic scenario where the interarrival
times are correlated: A two-state Markov chain $Z(n)$ alternates
between \emph{active} and \emph{inactive} periods; while in the
active state, exponentially distributed interarrival times are
generated with parameter $\lambda_{\text{act}}$, and the chain
turns inactive with probability $p>0$. In the inactive state,
\emph{one} interarrival time (exponentially distributed, parameter
$\lambda_{\text{inact}}<\lambda_{\text{act}}$) is generated, and
the chain jumps back to the active state (see
Figure~\ref{fig:markovchain}). Formally, let
\begin{equation*}
 t_{i,\text{act}}\sim\text{Exp}(\lambda_{\text{act}})~,\quad t_{i,\text{iact}}\sim\text{Exp}(\lambda_{\text{iact}})
\end{equation*}
be i.i.d. random variables and define the sequence of interarrival
times $t_i$ by
\begin{equation*}
 t_i:=t_{i,Z(i)}~.
\end{equation*}

\begin{figure}[t]
\begin{center}
\tikzstyle{every state}=[circle,fill=black!2,minimum
size=25pt,inner sep=0pt]
\begin{tikzpicture}[->,>=stealth',shorten >=1pt,auto,node distance=1.3cm,semithick]
      \node[state]                               (0) {$\text{iact}$};
      \node[state,right=of 0]                    (1) {$\text{act}$};
      \node[draw=none, below=of 0, yshift=0.2cm] (2) {$ $};
      \node[draw=none, below=of 1, yshift=0.2cm] (3) {$ $};
  \draw
  (1)   edge[loop right] node {$1-p$}   (1)
  (0)   edge[bend left=33] node {$1$}   (1)
  (1)   edge[bend left=33] node {$p$}   (0)
  (0)   edge node {$\lambda_{\text{iact}}$}   (2)
  (1)   edge node {$\lambda_{\text{act}}$}   (3);
\end{tikzpicture}
\caption{Two-state Markov chain $Z(n)$} \label{fig:markovchain}
\end{center}
\end{figure}
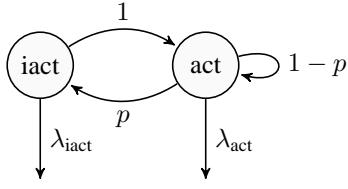

The steady state distribution $\pi$ of the Markov chain is given by
\begin{equation*}
 \pi_{\text{act}}=\frac{1}{1+p}~,\quad\text{and}\quad\pi_{\text{iact}}=\frac{p}{1+p}~,
\end{equation*}
such that for the average of the interarrival times holds
\begin{equation}
\label{eq:averaget}
 \E[t_i] = \left(\lambda_{\text{act}}^{-1} + p\lambda_{\text{iact}}^{-1}\right)\big/\left(1+p\right)
\end{equation}

Note that the transition matrix of $Z(n)$ is given by:
\begin{equation*}
  T:=\begin{pmatrix}0 & 1\\p & 1-p\end{pmatrix}~.
 \end{equation*}

In order to state the main result of this section, we need the
following transform of matrix $T$:
\begin{definition}
 For $0\le\theta<\lambda_{\text{iact}}$, let $T_{\theta}$ denote the following matrix:
 \begin{equation*}
  T_{\theta}:=\begin{pmatrix}0 & \frac{\lambda_{\text{act}}}{\lambda_{\text{act}}+\theta}\\p\frac{\lambda_{\text{iact}}}{\lambda_{\text{iact}}+\theta} & (1-p)\frac{\lambda_{\text{act}}}{\lambda_{\text{act}}+\theta}\end{pmatrix}~.
 \end{equation*}
  Further, let $\xi(\theta)$ denote the spectral radius of $T_{\theta}$, and $h=(h_{\text{act}},h_{\text{iact}})$ be a corresponding eigenvector.
\end{definition}
Note that $T_{\theta}$ is an exponential transform of $T$ which
has the Laplacians of the respective arrival times as an
additional factor in each column. In particular, with $\theta=0$
we recover the transition matrix itself, i.e., $T_{0}=T$.

The following Theorem is the analogous result to
Theorem~\ref{tm:iair} (note that the service times $x_{i,j}$ are
still assumed to be i.i.d.):
\begin{theorem}
\label{tm:mair}
 Let $1\le k\le K$ and $\theta_{\text{mkv}}$ be defined by
 \begin{equation*}
  \theta_{\text{mkv}}:=\sup\left\{\theta\ge 0\;\middle|\;\E\left[e^{\theta\min_{j\le k}\{x_{i,j}\}}\right]\xi^{\frac{K}{k}}(\theta)\le
  1\right\}~.
 \end{equation*}
 Then, for the system with replication to $k$ out of $K$ servers, the following bound on the response time holds for all $\sigma>0$:
 \begin{equation*}
  \P(r\ge\sigma)\le \E\left[e^{\theta_{\text{mkv}}\min_{j\le
  k}\{x_{i,j}\}}\right]e^{-\theta_{\text{mkv}}\sigma}~.
 \end{equation*}
\end{theorem}
\begin{proof}
  Proceeding similarly as in the proof of Theorem~\ref{tm:iair}, define the process $M(n)$ by
 \begin{equation*}
  M(n):=h_{Z(n\frac{K}{k}-1)}e^{\theta_{\text{mkv}}\left(\sum_{i=1}^{n} \tilde{x}_i-\sum_{i=1}^{n-1} \tilde{t}_i\right)}~.
 \end{equation*}
 $M(n)$ is a martingale: By induction over $\frac{K}{k}-1$ one shows that:
 \begin{align*}
  \E & \left[e^{-\theta_{\text{mkv}}\tilde{t}_{n+1}}\;\middle|\;Z\left(n\frac{K}{k}-1\right)\right]\\
  & = \left(T_{\theta_{\text{mkv}}}^{\frac{K}{k}}\right)_{Z((n\frac{K}{k}-1)),\text{iact}} + \left(T_{\theta_{\text{mkv}}}^{\frac{n}{l}}\right)_{Z(n\frac{K}{k}-1),\text{act}}~.
 \end{align*}
  Now:
 \begin{align*}
  \E & \left[h_{Z((n+1)\frac{K}{k}-1)}e^{\theta_{\text{mkv}}\left(\tilde{x}_{n+1}-\tilde{t}_{n}\right)}\;\middle|\; Z\left(n\frac{K}{k}-1\right)=\text{act}\right]\\
  & =  \E\left[e^{\theta_{\text{mkv}}\min_{j\le k}\{x_{n,j}\}}\right]\left(T_{\theta_{\text{mkv}}}^{\frac{K}{k}}h\right)_{\text{act}}\\
  & = \E\left[e^{\theta_{\text{mkv}}\min_{j\le k}\{x_{n+1,j}\}}\right]\;\xi^{\frac{K}{k}}(\theta_{\text{mkv}})h_{\text{act}}\\
  & = h_{\text{act}}~,
 \end{align*}
 and similarly one obtains:
 \begin{equation*}
  \E\left[h_{Z((n+1)\frac{K}{k}-1)}e^{\theta_{\text{mkv}}\left(\tilde{x}_{n+1}-\tilde{t}_{n}\right)}\;\middle|\; Z\left(n\frac{K}{k}-1\right)=\text{iact}\right] = h_{\text{iact}}~,
 \end{equation*}
 so that:
 \begin{equation*}
  \E\left[h_{Z((n+1)\frac{K}{k}-1)}e^{\theta_{\text{mkv}}\left(\tilde{x}_{n+1}-\tilde{t}_{n}\right)}\;\middle|\; Z\left(n\frac{K}{k}-1\right)\right] = h_{Z(n)}~.
 \end{equation*}
 Now multiply both sides by $e^{\theta_{\text{mkv}}\left(\sum_{i=1}^{n} \min_{j\le k}\{x_{i,j}\}-\sum_{i=1}^{n-1} t_i\right)}$. The proof completes along the same kind of lines as in the proof of Theorem~\ref{tm:iair}.
\end{proof}

\subsection{Independent Arrivals, Correlated Replication}\label{subsec:iacr}
We now address the more realistic scenario when the replicas
$x_{i,j}$ are no longer independent; we consider the following
correlation model (from \cite{Joshi14}):
\begin{equation}
\label{eq:cormod}
 x_{i,j}=\delta y_i + \left(1-\delta\right) y_{i,j}~,
\end{equation}
where the random variables $y_i$ and $y_{i,j}$ are i.i.d. Here,
the parameter $\delta$ describes the degree of correlation amongst
the replicas: $\delta=0$ corresponds to the i.i.d. case from
Section~\ref{subsec:iair}, whereas for $\delta=1$ the $K$ servers
are entirely synchronized so that no replication gain is achieved.

For simplicity the interarrival times $t_i$ are first assumed to
be i.i.d. as in Section~\ref{subsec:iair}.
\begin{theorem}
\label{tm:iacr}
 Let $\theta_{\text{cor}}$ be defined by
 \begin{align*}
  \theta_{\text{cor}}:=\sup\bigg\{\theta\ge 0\;\bigg|\;\E\left[e^{\theta\delta y_i}\right] & \E\left[e^{\theta\left(1-\delta\right)\min_{j\le k}\{y_{i,j}\}}\right]\\
  &\E\left[e^{-\theta t_i}\right]^{\frac{K}{k}}\le 1\bigg\}~.
 \end{align*}
Then the following bound on the response time holds for all
$\sigma\ge 0$:
 \begin{equation*}
  \P(r\ge\sigma)\le \E\left[e^{\delta\theta_{\text{cor}} y_i}\right]\E\left[e^{\left(1-\delta\right)\theta_{\text{cor}}\min_{j\le k}\{y_{i,j}\}}\right]e^{-\theta_{\text{cor}}\sigma}~.
 \end{equation*}
\end{theorem}
\begin{proof}
 Entirely analogous to the proof of Theorem~\ref{tm:iair}.
\end{proof}

To illustrate the impact of the correlation parameter $\delta$ we
consider the special case when $y_i$ and $y_{i,j}$ are
exponentially distributed with parameter $\mu$. Clearly,
\begin{equation*}
 \min_{j\le k}\{y_{i,j}\}\sim\text{Exp}(k\mu)~,
\end{equation*}
so that $\theta_{\text{cor}}>0$ is the solution of
\begin{equation*}
 \frac{\mu}{\mu-\delta\theta}\;\frac{k\mu}{k\mu-\left(1-\delta\right)\theta}\;\frac{\lambda}{\lambda+\theta}=1~.
\end{equation*}

\begin{figure}[t]
\begin{center}
\includegraphics[width=0.48\textwidth]{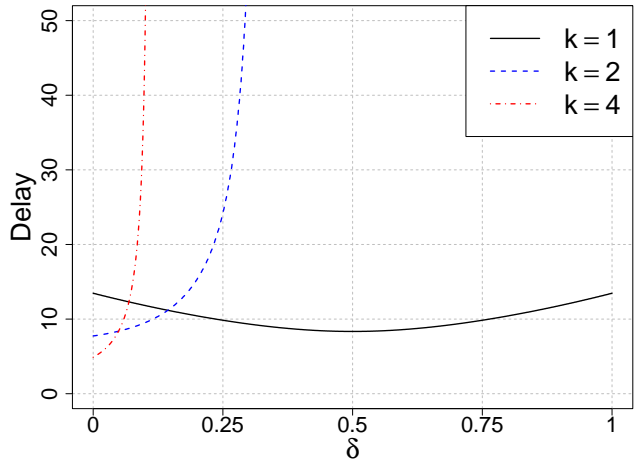}
\caption{Delay for the $99\%$-percentile as a function of the
degree of correlation $\delta$ ($\lambda=4*0.75$, $\mu=1$, $K=4$,
    $k=1,2,4$)} \label{fig:deltaplotdelay}
\end{center}
\end{figure}

Further, Figure~\ref{fig:deltaplotdelay} illustrates the
$99\%$-percentile of the delay as a function of the degree of
correlation $\delta$ for several numbers of replicas $k$. Strictly
from the point of view of the stability region, as it was also
considered in Section~\ref{sec:stability}, we observe that
replication (both $k=2$ and $k=4$) is detrimental as the
corresponding systems quickly become unstable. In contrast, from
the point of view of delays, replication can be beneficial within
a subset of the corresponding stability region notwithstanding its
strict inclusion in the stability region of the non-replicated
system. This fundamental observation can be intuitively explained
in that for larger values of the degree of correlation $\delta$,
the servers become more synchronized and consequently no
significant \emph{replication gain} can be achieved; a further
follow-up discussion concerning a convergence result depending on
$\delta$ will be given in Section~\ref{sec:forkjoin}. As a side
remark, the symmetry in the delay for $k=1$ is due to the
underlying Erlang distribution, which minimizes its variance at
$\delta=.5$.

\subsection{Markovian Arrivals, Correlated Replication}\label{subsec:macr}

\begin{figure*}[t]
\centering
    \subfloat[\label{fig:boundiid}{Poisson arrivals, independent exponential replication (Theorem~\ref{tm:iair}, $\lambda=4\ast 0.75$, $\mu=1$)}]{%
      \includegraphics[width=0.45\textwidth]{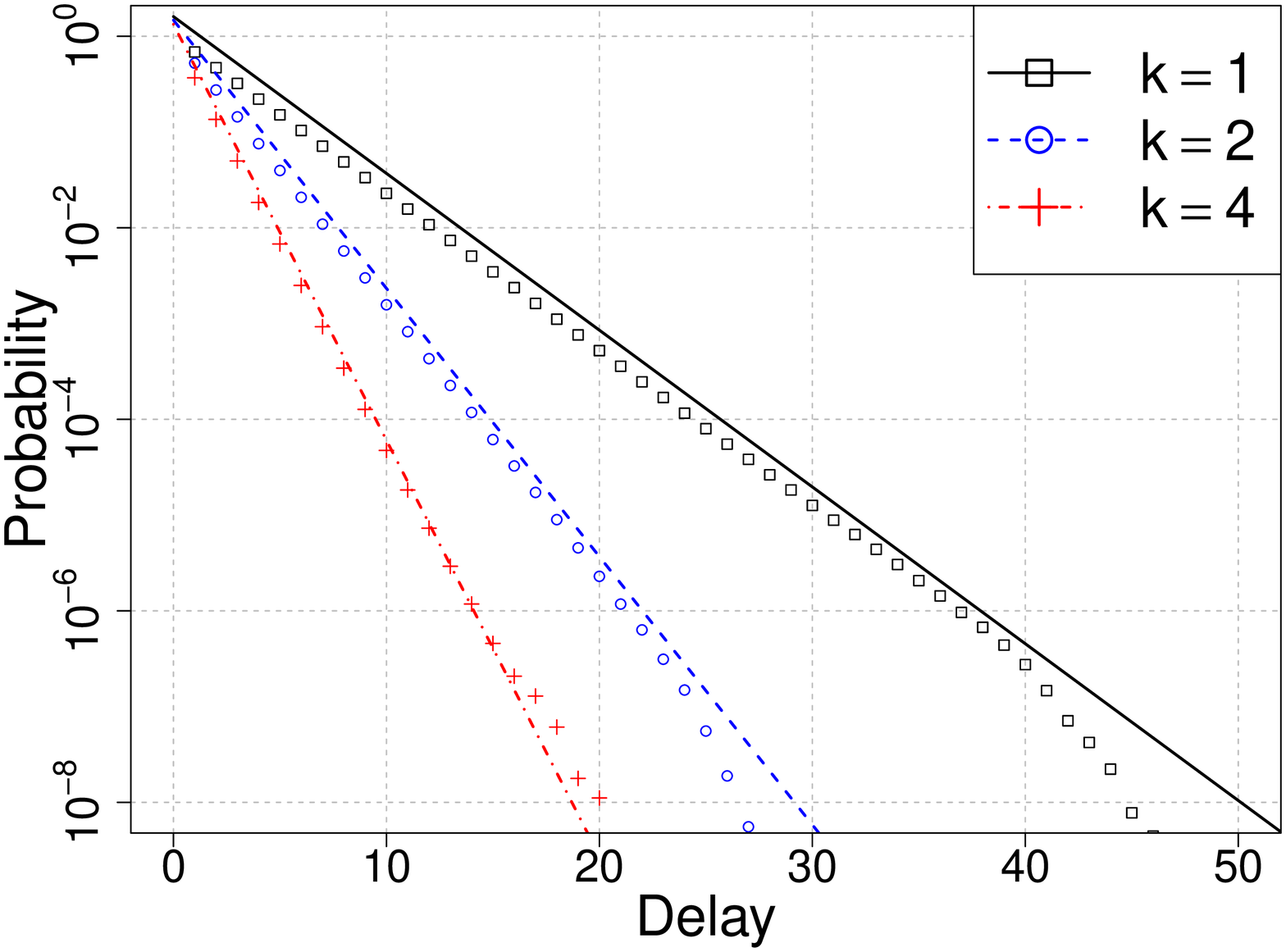}
    }
    \hspace{40pt}
    \subfloat[\label{fig:boundcor}{Poisson arrivals, correlated exponential replication (Theorem~\ref{tm:iacr}, $\lambda=4\ast 0.75$, $\delta=0.5$, $\mu':=\delta k+\left(1-\delta\right)$)}]{%
      \includegraphics[width=0.45\textwidth]{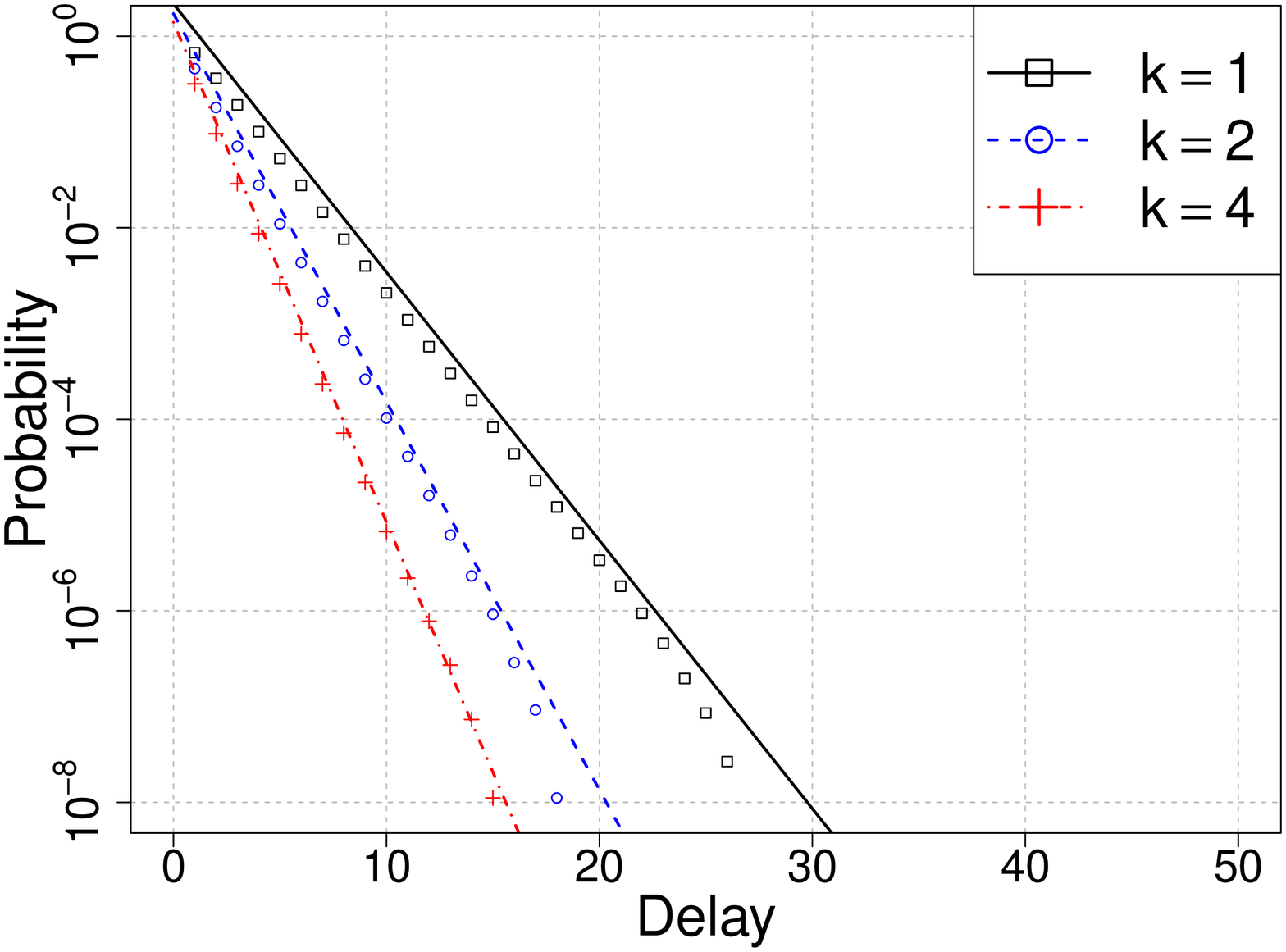}
    }

    \subfloat[\label{fig:boundmarkov}{Markovian arrivals, independent exponential replication (Theorem~\ref{tm:mair}, $p=0.1$, $\lambda_{\text{iact}}=0.3$, $\lambda_{\text{act}}=30$, $\mu=1$)}]{%
      \includegraphics[width=0.45\textwidth]{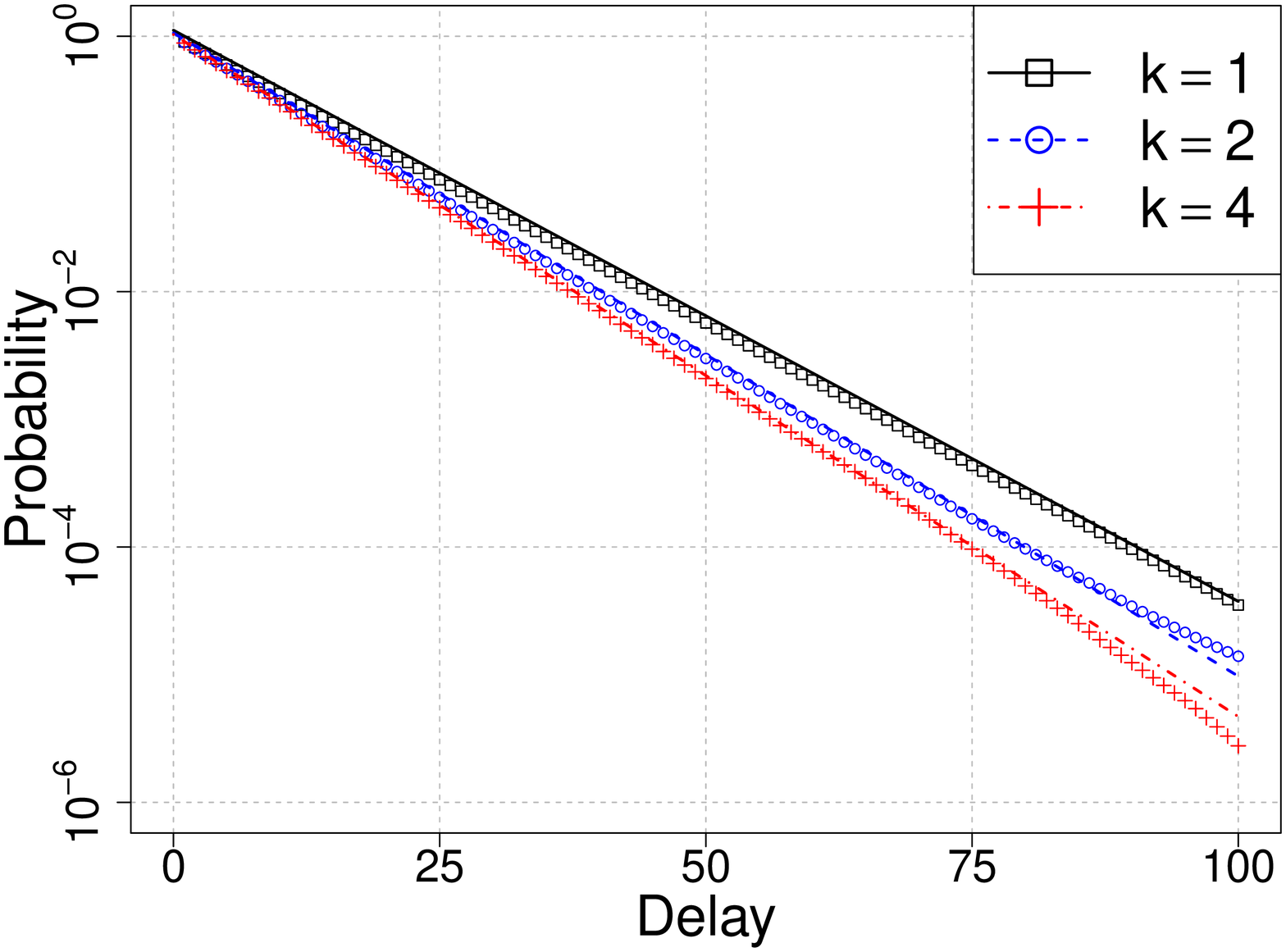}
    }
    \hspace{40pt}
    \subfloat[\label{fig:boundcormar}{Markovian arrivals, correlated exponential replication (Theorem~\ref{tm:macr}, $p=0.1$, $\lambda_{\text{iact}}=0.3$, $\lambda_{\text{act}}=30$, $\mu=1$, $\delta=0.5$, $\mu':=\delta k+\left(1-\delta\right)$)}]{%
      \includegraphics[width=0.45\textwidth]{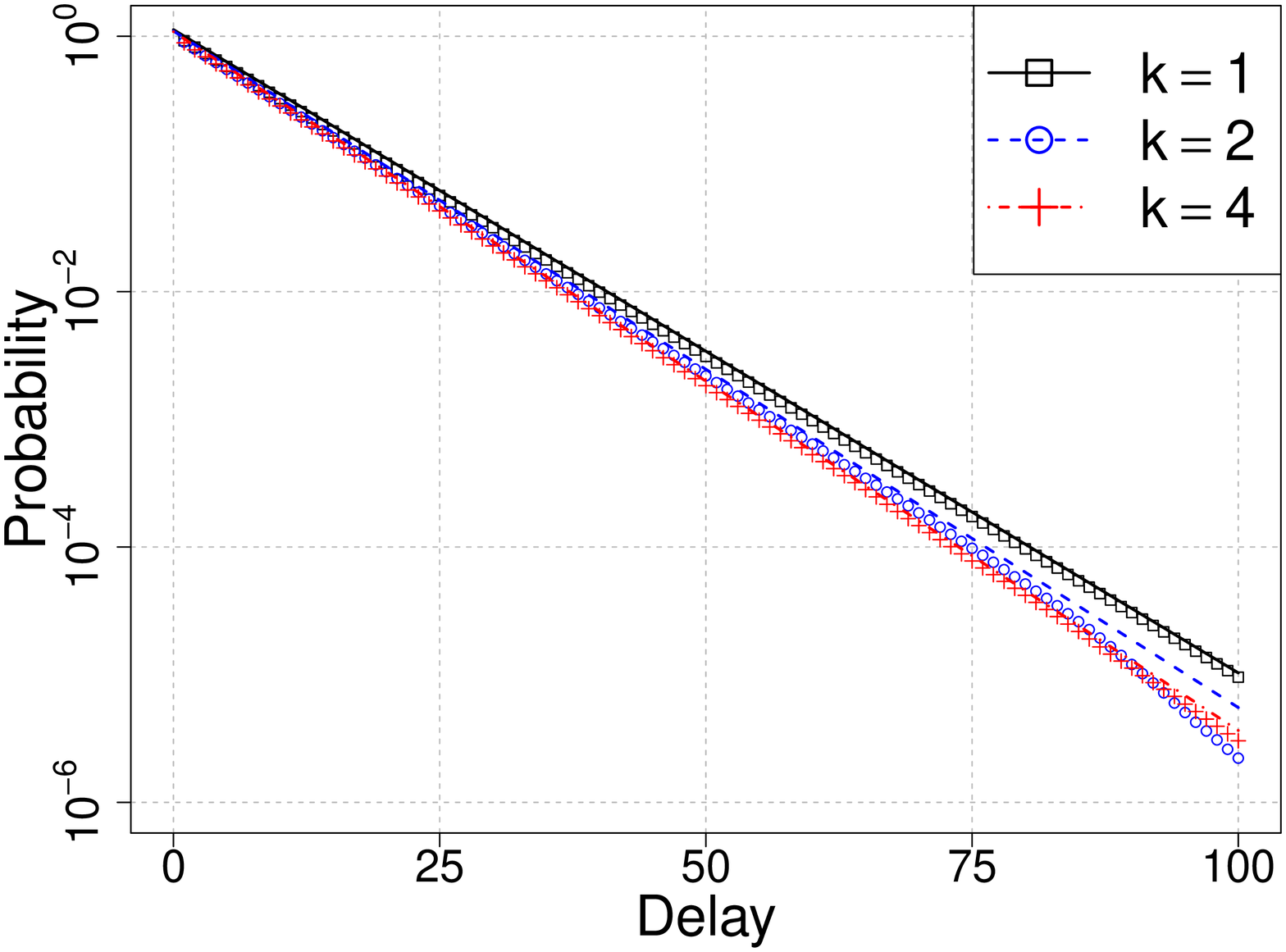}
    }
    \caption{Stochastic bounds vs. simulation results accounting for $10^9$ packets ($K=4$, $\rho=0.75$, $k=1,2,4$)}
    \label{fig:markov}
  \end{figure*}

We briefly state the results for the combination of the scenario
from Sections~\ref{subsec:mair} and \ref{subsec:iacr}:

\begin{theorem}
\label{tm:macr}
 With the same notation as in Sections~\ref{subsec:mair} and \ref{subsec:iacr}, let $\theta_{\text{mkv,cor}}$ be defined by
 \begin{align*}
  \theta_{\text{mkv,cor}}:=\sup\bigg\{\theta\ge 0\;\Big|\; & \E\left[e^{\theta\delta y_i}\right]
   \E\left[e^{\theta\left(1-\delta\right)\min_{j\le k}\{y_{i,j}\}}\right]\\
   &\xi^{\frac{K}{l}}(\theta)\le 1\bigg\}~.
 \end{align*}
 Then the following bound on the response time holds for all $\sigma\ge 0$:
 \begin{equation*}
  \P(r\ge\sigma)\le \E\left[e^{\delta\theta_{\text{mkv,cor}} y_i}\right]\E\left[e^{\left(1-\delta\right)\theta_{\text{mkv,cor}}\min_{j\le k}\{y_{i,j}\}}\right]e^{-\theta_{\text{mkv,cor}}\sigma}~.
 \end{equation*}
\end{theorem}
\begin{proof}
 Entirely analoguous to the proofs of Theorems~\ref{tm:iair} and \ref{tm:mair}.
\end{proof}

To numerically compare our stochastic bounds from
Theorems~\ref{tm:iair},~\ref{tm:iacr},~\ref{tm:mair},
and~\ref{tm:macr} to simulation results we refer to
Figures~\ref{fig:markov}(a)-(d), respectively. In all four
scenarios, addressing combinations of independent/correlated
arrivals and replications, jobs are replicated to $k=1,2,4$ out of
a total number of $K=4$ servers. The parameters of the respective
models are chosen such that the (server) utilization remains
constant, i.e., $\rho=0.75$. In particular, in
Figure~\ref{fig:boundiid}, both the interarrival- and service
times are exponentially distributed with parameters $\lambda=4\ast
0.75$ and $\mu=1$. In Figure~\ref{fig:boundcor}, the interarrival
times are again exponential with $\lambda=4\ast 0.75$, the
correlation factor is $\delta=0.5$, whereas the components $y_i$
and $y_{i,j}$ of the service times $x_{i,j}$ from
Eq.~(\ref{eq:cormod}) are exponential with parameter
\begin{equation*} \mu' := \delta + \left(1-\delta\right)\big/ k~,
\end{equation*} such that $\E[x_{i,j}]=1$. In
Figure~\ref{fig:boundmarkov}, the parameters for the Markov chain
are $p=0.1$, $\lambda_{\text{act}}=30$,
$\lambda_{\text{iact}}=0.3$, whereas the services times are
exponential with parameter $\mu=1$. According to
Eq.~(\ref{eq:averaget}) the average of the interarrival times is
$E[t_i]=1/3$, such that $\rho=0.75$. Finally, in
Figure~\ref{fig:boundcormar}, the parameters for the service times
from Figure~\ref{fig:boundcor} are combined with the parameters
for the interarrival times from Figure~\ref{fig:boundmarkov}. We
remark that in all four scenarios the stochastic bounds from
Theorems~\ref{tm:iair},~\ref{tm:iacr},~\ref{tm:mair},
and~\ref{tm:macr} are remarkably accurate.

\section{Applications}\label{sec:app}
In this section we present two practical applications of our
theoretical framework. The first concerns integrating replication
with a fork-join queueing model; a major outcome is the
construction of an intuitive class of assignment policies which
can fundamentally improve response times. The second investigates
the analytical tradeoff between resource usage and response times,
an issue which was subject to several measurement studies
involving Google and Bing traces.

\subsection{Fork-Join with Replication (FJR)}\label{sec:forkjoin}

In this section we consider replication in the context of a
fork-join (FJ) queueing system. In a FJ system, arriving jobs are
split into $K$ different tasks which are mapped to $K$ servers to
be processed independently. A job is considered finished once
\emph{all} of its corresponding tasks have finished. We consider
the special case of a \emph{blocking} system whereby jobs cannot
be forked before all of the tasks of the previous job have left
the system (this mode is in particular characteristic to Hadoop,
through a particular coordination service~\cite{White09}).

The obvious drawback of this blocking model is that it is no
longer work-conserving: servers can become idle once some but not
all tasks of one job are complete. Moreover, the stability
condition of the system becomes a function of the number of
servers.

Consider for instance the case of Poisson arrivals with rate
$\lambda$ and exponential and identically distributed service
times $x_i$, $i=1,\dots,K$, with rate $\mu$. As the distribution
of the maximum of i.i.d. exponential random variables satisfies
$\max_{i=1}^K
x_i=_{\mathcal{D}}\sum_{i=1}^K\frac{x_i}{i}$~\cite{renyi53}, the
stability condition is roughly
\begin{equation}
\frac{\lambda}{\mu}\ln K<1~.\label{eq:fjstab}
\end{equation}

To overcome the issue of decaying stability regions (in the number
of servers $K$) we propose the following task assignment policy
which suitably triggers replicas on top of the standard FJ model.

\noindent{\textit{Policy FJR (Fork-Join with Replication)}: Once a
server finishes its task, it immediately replicates a remaining
task from another running server. When either the original task or
one of its replica has finished, the others are immediately
purged.}

FJR can be regarded as a concrete implementation of backup-tasks
in MapReduce (which is not explicitly presented in the original
MapReduce description~\cite{Dean08}). Our policy is quite flexible
in that the executing task to be replicated can be chosen randomly
(yet independently of the current state); moreover, as multiple
servers can become idle at the same time (due to the underlying
purging model), each can replicate any executing tasks.
Intuitively, this flexibility is due to the underlying assumption
of exponentially distributed and independent service times.

The main result of the FJR policy is the following:

\begin{theorem}
\label{tm:fjrstab}
 The \emph{overall} service time $x$ of jobs processed by FJR follows an $\text{Erlang}(K,K\mu)$-distribution. Consequently, the corresponding stability condition is
  \begin{equation*}
  \frac{\lambda}{\mu}<1~.
 \end{equation*}
\end{theorem}

\begin{proof}
Let $y_1 < y_2 < \ldots < y_K$ denote the times where the tasks
(original or replica) finish (see Figure~\ref{fig:forkjoin}).
Obiviously, it holds $x=y_K$. We first show (with the convention
$y_0\equiv 0$) that the family
\begin{equation*}
\left\{y_i-y_{i-1} \; \middle| \; i\ge 1\right\}
\end{equation*}
is independent and identically exponentially distributed with
parameter $K \mu$.

For $i=1$, this follows directly from the well known fact that the
minimum over $K$ independent, exponential random variables with
rate $\mu$ is exponentially distributed with rate $K\mu$.

Now, suppose $1\le l\le K$ tasks finish, or are purged, at time
$y_i$. Denote by $z_1,\ldots,z_l$ the corresponding service times
of the respective replicas starting at $y_i$. For the remaining
$K-l$ servers, denote by $z_{l+1},\ldots,z_K$ the service times of
the current tasks and by $s_{l+1},\ldots,s_{K}$ the length of time
they started before $y_i$. Now we can write
\begin{align*}
 y_{i+1}-y_i = \min \big\{ z_1, \ldots, z_l, z_{l+1}-s_{l+1}, \ldots, z_{K}-s_K\big| z_{l+1}-s_{l+1},\ldots,z_{K}-s_K>0\big\}~.
\end{align*}
Note that the family $\{z_1,\ldots,z_K\}$ is independent from one
another and from $\{s_{l+1},\ldots s_K\}$.

\begin{figure}[t]
\begin{center}
\begin{tikzpicture}[very thick]
    \draw (0,3) node[left,xshift=-0.2cm]{$\textrm{Server 1: }$};
    \draw (0,2) node[left,xshift=-0.2cm]{$\textrm{Server 2: }$};
    \draw (0,1) node[left,xshift=-0.2cm]{$\textrm{Server 3: }$};
    \draw (0,0) node[left,xshift=-0.2cm]{$\textrm{Server 4: }$};

    \draw[color=black] (0,3) -- (1.3,3);
    \draw[color=red] (1.3,3.05) -- (2.7,3.05);
    \draw[color=red, dotted] (2.7,3.05) -- (3.2,3.05);
    \draw[color=blue] (2.7,3.1) -- (3.7,3.1);
    \draw[color=black] (0,3.11) -- (0,2.89);

    \draw[color=red] (0,2) -- (2.7,2);
    \draw[color=red, dotted] (2.7,2) -- (4.3,2);
    \draw[color=blue] (2.7,2.05) -- (3.7,2.05);
    \draw[color=blue, dotted] (3.7,2.05) -- (4.7,2.05);
    \draw[color=black] (0,2.11) -- (0,1.89);

    \draw[color=teal] (0,1) -- (2,1);
    \draw[color=red] (2,1.05) -- (2.7,1.05);
    \draw[color=blue] (2.7,1.1) -- (3.7,1.1);
    \draw[color=blue, dotted] (2.7,1.1) -- (4.1,1.1);
    \draw[color=black] (0,1.11) -- (0,0.89);

    \draw[color=blue] (0,0) -- (3.7,0);
    \draw[color=blue, dotted] (3.7,0) -- (5,0);
    \draw[color=black] (0,0.11) -- (0,-0.11);

    \draw[color=gray,dashed, thick] (0,3.5) -- (0,-0.5);
    \draw (0,-0.5) node[below]{$y_0$};
    \draw[color=gray,dashed, thick] (1.3,3.5) -- (1.3,-0.5);
    \draw (1.3,-0.5) node[below]{$y_1$};
    \draw[color=gray,dashed, thick] (2,3.5) -- (2,-0.5);
    \draw (2,-0.5) node[below]{$y_2$};
    \draw[color=gray,dashed, thick] (2.7,3.5) -- (2.7,-0.5);
    \draw (2.7,-0.5) node[below]{$y_3$};
    \draw[color=gray,dashed, thick] (3.7,3.5) -- (3.7,-0.5);
    \draw (3.7,-0.5) node[below]{$y_4$};
    \end{tikzpicture}
 \caption{FJR policy; different colors denote different tasks, dotted lines indicate tasks which have been purged.}\label{fig:forkjoin}
\end{center}
\end{figure}
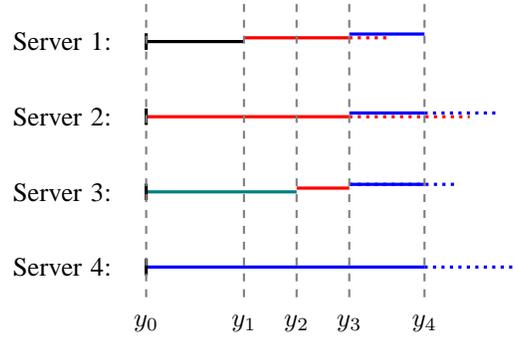

Now, with
\begin{equation*}
A:=\{z_{l+1}-s_{l+1},\ldots,z_{K}-s_K>0\}~,
\end{equation*}
$\vec{s}:=\left(s_{l+1},\ldots,s_{K}\right)$, and $f(.)$ the
common density of $\vec{s}$:
\begin{align*}
 \P (y_{i+1}-y_i\ge\sigma) & = \P\big(\min\left\{z_1,\ldots,z_l,z_{l+1}-s_{l+1}, \ldots, z_{K}-s_K\right\}\ge\sigma\big|A\big)\\
 & = e^{-l\mu\sigma}\int e^{-\mu\left(\sum_{j=l+1}^K\sigma+s_j\right)} f(\vec{s})d\vec{s}\Big/\P(A)\\
 & = e^{-K\mu\sigma}\int e^{-\mu\sum_{j=l+1}^K s_j} f(\vec{s})d\vec{s}\Big/\P(A)\\
 & = e^{-K\mu\sigma}\int \P(z_{l+1}>s_{l+1},\ldots,z_{K}>s_{K}) f(\vec{s})d\vec{s}\Big/\P(A)\\
 & = e^{-K\mu\sigma}~,
\end{align*}

so that $y_i-y_{i-1}$ is exponentially distributed for any $1\le
i\le K$. It follows that
\begin{equation*}
 x=y_K=\sum_{i=1}^K y_i-y_{i-1}
\end{equation*}
has an Erlang distribution with parameters $K$ and $K\mu$.
Therefore $E[x]=\frac{1}{\mu}$, which completes the proof.
\end{proof}

\begin{figure}[ht]
\begin{center}
\includegraphics[width=0.48\textwidth]{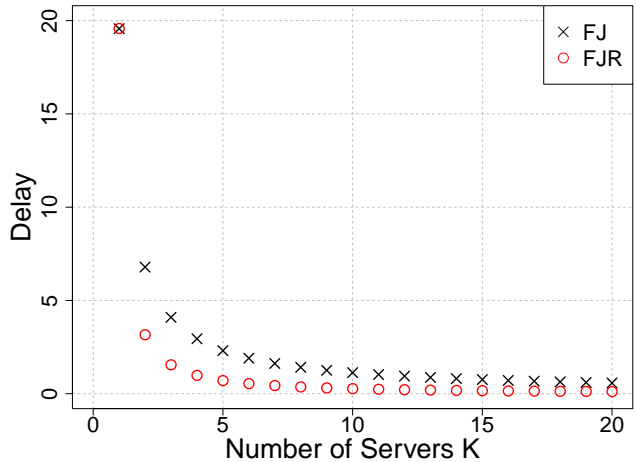}
\caption{Improving the $99\%$-percentile of delays in FJ systems
by replication} \label{fig:fjcompare}
\end{center}
\end{figure}

It is evident that the stability region of FJR improves the
stability region of the standard FJ queueing model (given in
Eq.~(\ref{eq:fjstab})) by a logarithmic factor. To further
visualize the numerical difference between FJR and FJ in the
actual delays, we first observe that the response time
distribution can be expressed as in Theorem~\ref{tm:iair} by
letting $k=K$ and replacing the `$\min$' by a `$\max$' (see
also~\cite{RiPoCi15} for explicit results).
Figure~\ref{fig:fjcompare} shows the $99^{\textrm{th}}$ percentile
of the delays as a function of $K$ ($\mu=1$ and Poisson arrivals
with rate such $\rho=0.75$ when $K=1$; the utilization
consequently decays for larger $K$). The numerical benefit of FJR
is that it roughly halves the FJ delays.

While the fundamental improvements achieved by the FJR policy,
relative to the standard FJ model, are remarkable, we point out
that they are mainly due to the exponential and independence
assumptions on the triggered replicas. Unfortunately, a clean
analysis in the case of correlated replicas (even of the form
$(1-\delta)x_i+\delta x$, with $x$ and $x_i$'s being exponentially
distributed) appears prohibitive. For this reason, we resort to
simulations to illustrate that the benefits of FJR (proven in the
ideal i.i.d. and exponential case) carry over to more practical
scenarios with correlated replicas.

Concretely, Figure~\ref{fig:fjrconv} shows the bounds on the delay
distributions for FJ and three FJR scenarios, depending on the
degree of correlation $\delta$ (the service times of an original
and its replicated tasks are $(1-\delta)x_i+\delta x$, with $x$
and $x_i$'s being exponentially distributed with rate $\mu=1$;
Poisson arrivals such that the utilization for FJ is $\rho=0.9$
(the corresponding utilizations for FJR are not analytically
determined)). The figure essentially illustrates the convergence
of FJR to FJ; we remark in particular that FJ is invariant to
$\delta$, whereas $FJR$ behaves identically as $FJ$ when
$\delta=1$ (i.e., when the replicas are identical to the
originals).

\subsection{Resource Usage vs. Response Times} For the second application we investigate the analytical
tradeoff between resource usage and response times under
replication. This application is motivated by empirical
observations from Google~\cite{Dean13} and Bing~\cite{Jalaparti13}
traces that a slight increase in the resource budget may yield
substantial reductions of the upper quantiles of response times.
For example,~\cite{Jalaparti13} reports that the
$99^{\textrm{th}}$ percentile of the delay improves by as much as
$40\%$ under a $5\%$ increase of the resource budget. To
compensate for the inherent increase of resource usage under
replication, the schemes from \cite{Dean13,Jalaparti13} defer the
execution time of the replicas until the original request has been
outstanding for a given \textit{replication offset} $\Delta$.

\begin{figure}[h]
\begin{center}
\includegraphics[width=0.48\textwidth]{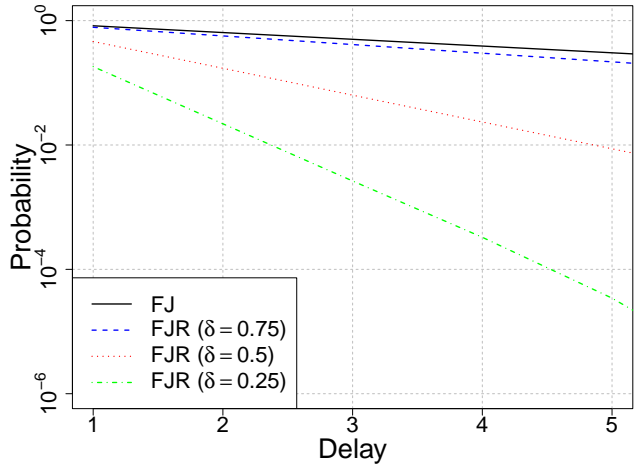}
\caption{Convergence of FJR to FJ in terms of the degree of
correlation $\delta$ ($K=4$).} \label{fig:fjrconv}
\end{center}
\end{figure}

\begin{figure}[t]
\begin{center}
\begin{tikzpicture}[thick]
    \draw (0,0) -- (2,0);
    \draw (0,-0.11) -- (0,0.11);
    \draw (2,-0.11) -- (2,0.11);
    \draw (2,0) -- (4,0);
    \draw (4,-0.11) -- (4,0.11);
    \draw (1,0) node[above]{$\Delta$};
    \draw (3,0) node[above]{$y$};
    \draw (0,1) -- (3,1);
    \draw (0,0.89) -- (0,1.11);
    \draw (3,0.89) -- (3,1.11);
    \draw (1.5,1) node[above]{$x$};
    \draw (0,0) node[left,xshift=-0.2cm]{$\textrm{Server 2: }$};
    \draw (0,1) node[left,xshift=-0.2cm]{$\textrm{Server 1: }$};
    \draw (4,0) node[right]{$\dots$};
    \draw (3,1) node[right]{$\dots$};
\end{tikzpicture}
 \caption{Replication with deferred execution times: a replica (at Server 2) may start no sooner than ($\Delta\geq0$) after the starting time of the original (at Server 1).}\label{fig:ReplicaDelta}
\end{center}
\end{figure}
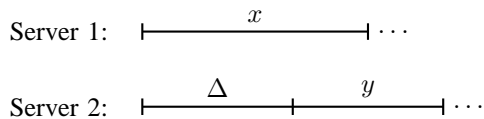

Consider a scenario with two servers. Jobs arrive with rate
$\lambda$ at the first server with interarrival times $t_i$ and
service times $x_i=_{\mathcal{D}} x$; if the processing time of a
job is larger than some fixed $\Delta$, then the job is replicated
at the second server with service times $y_i=_{\mathcal{D}} y$
(see Figure~\ref{fig:ReplicaDelta} for a time-line illustration of
a generic job with execution time $x$ and its replica, should
$x>\Delta$). Whenever either of the original job or its replica
finishes execution, the residual service time of the other is
cancelled (i.e., the purging replication model).

The utilization at the first server is thus given by
\begin{equation}
\rho_1=\lambda E\left[\min\{x,\Delta+y\}\right]~,\label{eq:rhos1}
\end{equation}
whereas the utilization at the second is
\begin{equation}
\rho_2=\lambda
E\left[\min\left\{|x-\Delta|,y\right\}\right]~.\label{eq:rhos2}
\end{equation}
We note that unlike previous models, where the utilization is
server independent, the current model is subject to different
server utilizations due to the lack of symmetry in dispatching the
load.

The measure for \textit{resource usage} is the total utilization
at the two servers and is denoted by $u$ to avoid confusion
\begin{equation*}
u:=\rho_1+\rho_2~.
\end{equation*}

\begin{figure*}[t]
\centering
    \subfloat[\label{fig:deltadelta25}{$\delta=0.25$}]{%
      \includegraphics[width=0.45\textwidth]{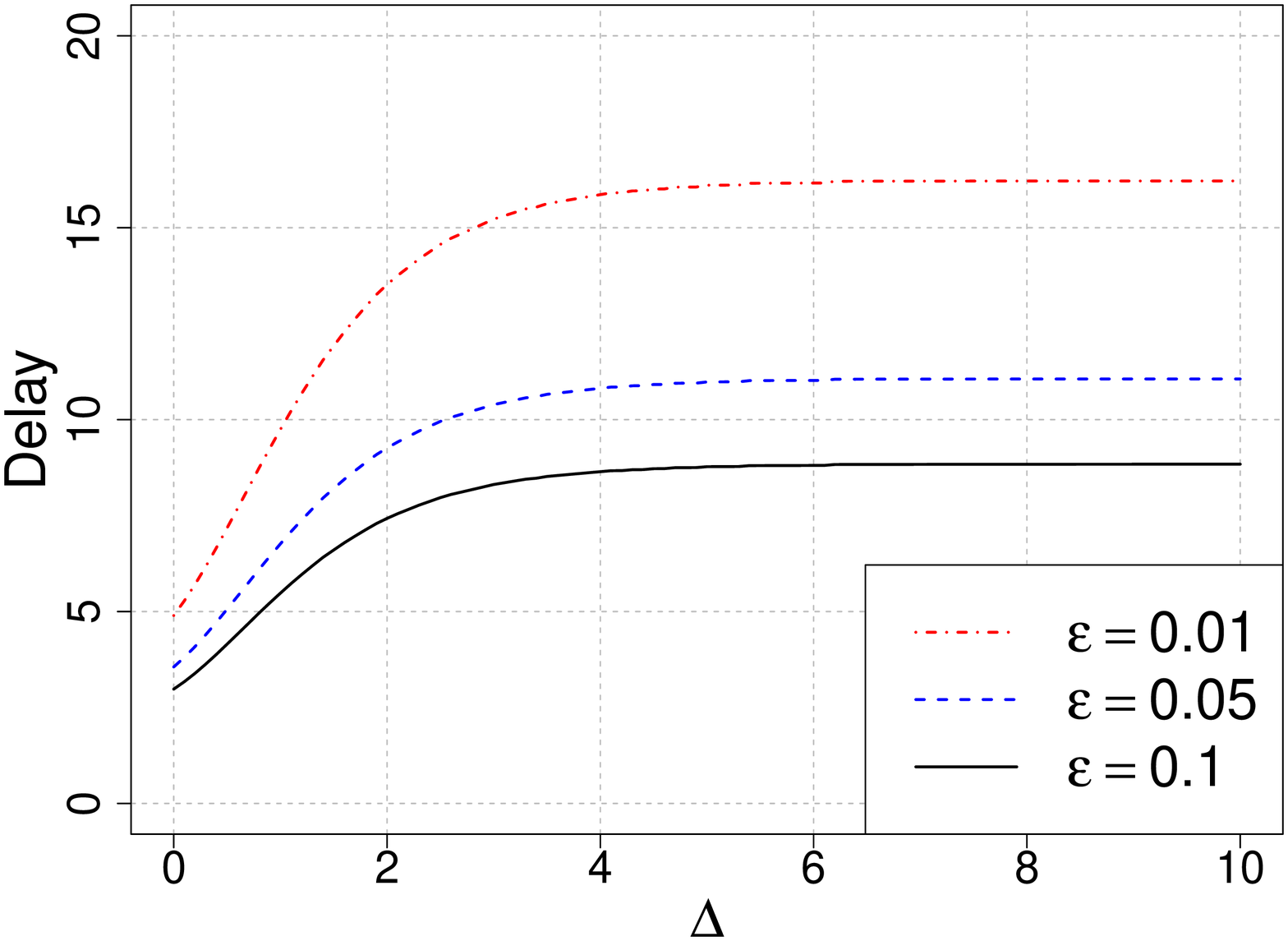}
    }
    \hspace{40pt}
    \subfloat[\label{fig:deltadelta75}{$\delta=0.75$}]{%
      \includegraphics[width=0.45\textwidth]{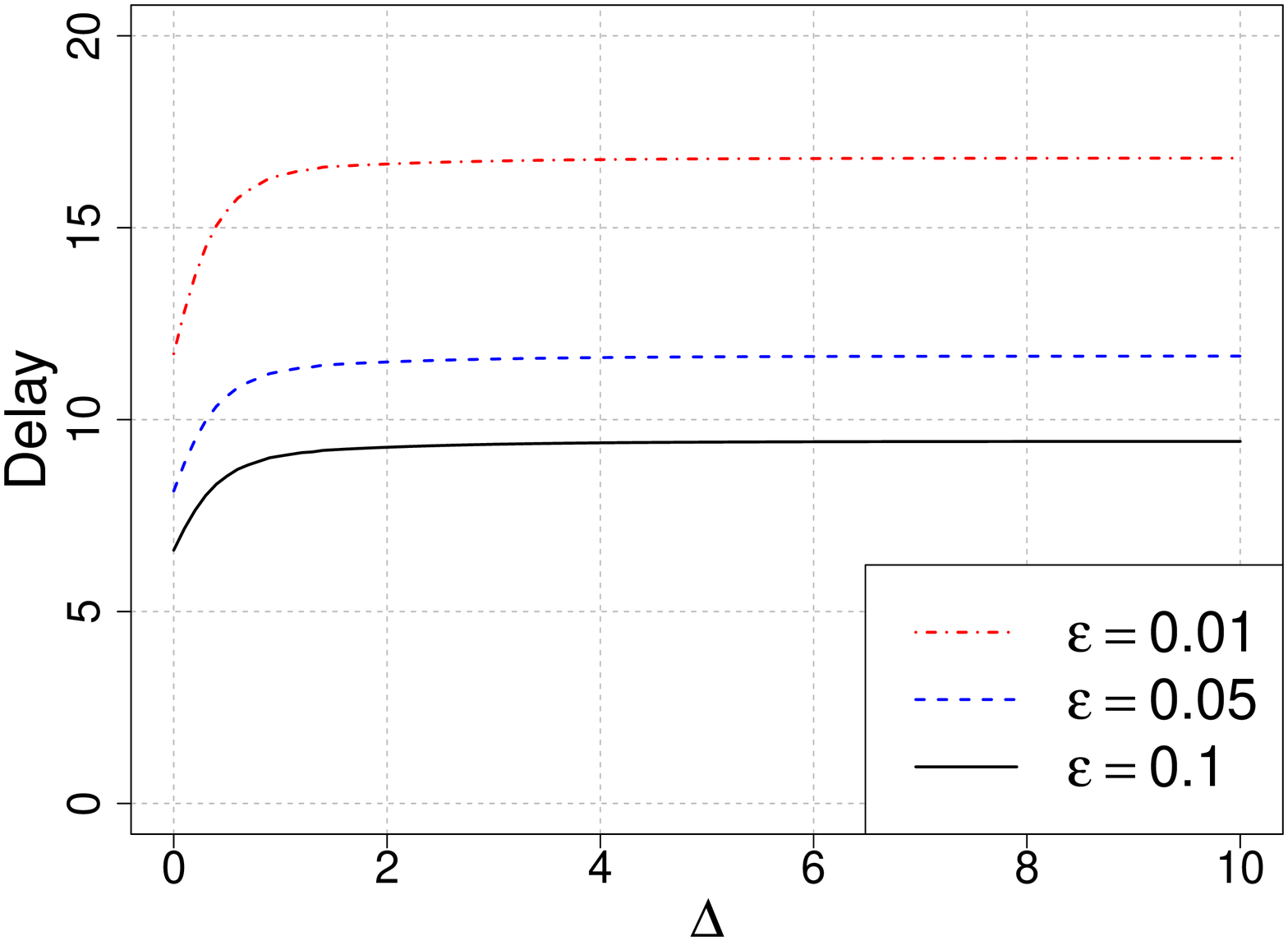}
    }
    \caption{Quantiles of the response time vs. the replication offset $\Delta$ ($\lambda=0.75$, $\mu=1$)}
    \label{fig:deltadelta}
\end{figure*}

Aiming for explicit results, we assume for convenience the
exponential service model, i.e., $x\sim\textrm{exp}(\mu)$ and
$y\sim\textrm{exp}(\mu)$, with $\mu=1$. Moreover, we consider both
the independent and correlated replication models.

\subsubsection{Independent Replication} Given the statistical
independence of $x_i$'s and $y_i$'s, straightforward computations
of integrals yield
\begin{eqnarray*}
\rho_1=\frac{\lambda}{\mu}-\frac{\lambda}{2\mu}e^{-\mu\Delta}~,\quad\textrm{and }\quad\rho_2=\frac{\lambda}{2\mu}e^{-\mu\Delta}~,
\end{eqnarray*}
which means that the resource usage $u=\frac{\lambda}{\mu}$ is
invariant to the choice of $\Delta$.

In turn, $\Delta$ can have a major impact on the response times:
for instance, if $\mu<\lambda<2\mu$ then the response times can be
either unbounded for sufficiently large values of $\Delta$, and in
particular when $\Delta=\infty$ (i.e., no replicas are executed),
or finite for some values of $\Delta$.

In fact, an immediate application of Theorem~\ref{tm:iair} yields
that the response time is non-decreasing in $\Delta$. Thus, the
optimal choice of $\Delta$, which minimizes both the resource
usage and the response times, is $\Delta=0$. The explanation for
the seemingly sharp contrast between this theoretical result and
the empirical results from~\cite{Dean13,Jalaparti13} is the
underlying independence assumption of the replication model.

\subsubsection{Correlated Replication} A non-trivial tradeoff
between resource usage and response times manifests itself under
the more realistic correlated replication model from Section
\ref{subsec:iacr}. The original and replica response times are
modelled by
\begin{equation*}
(1-\delta)x+\delta z~\textrm{and}~(1-\delta)y+\delta z~,
\end{equation*}
where $x,y,$ and $z$ are exponential with rate $\mu=1$. The
parameter $\delta$ sets the degree of correlation; in particular,
small values of $\delta$ indicate a small degree of correlation.

Rather tedious computations of integrals, due to several
conditions stemming from the absolute value operator in $\rho_2$,
yield the individual utilizations
\begin{eqnarray*}
\rho_1&=&\frac{\lambda}{\mu}\left(1-\frac{1-\delta}{2}e^{-\frac{\mu}{1-\delta}\Delta}\right)~\textrm{and}\\
\rho_2&=&\frac{\lambda}{\mu}\left(\frac{\delta^2}{2\delta-1}e^{-\frac{\mu}{\delta}\Delta}-\frac{1-\delta}{2(2\delta-1)}e^{\frac{\mu}{1-\delta}\Delta}\right)~,
\end{eqnarray*}
and further the resource usage
\begin{equation}
u=\frac{\lambda}{\mu}\left(1+\frac{\delta^2}{2\delta-1}e^{-\frac{\mu}{\delta}\Delta}-\frac{\delta(1-\delta)}{2\delta-1}e^{-\frac{\mu}{1-\delta}\Delta}\right)
\label{eq:rhocore}
\end{equation}
under the assumptions that $\delta\in(0,1)$ and $\delta\neq .5$.

To illustrate a quantitative tradeoff between resource usage
(Eq.~(\ref{eq:rhocore})) and response times
(Theorem~\ref{tm:iacr}), we refer to Figure~\ref{fig:deltadelta}
which shows the increase of the top percentiles of the response
times ($90^{\textrm{th}}$, $95^{\textrm{th}}$, and
$99^{\textrm{th}}$) as a function of the replication offset
$\Delta$. Both small ($\delta=0.25$) and high ($\delta=0.75$)
correlation degrees are considered; in Figure~\ref{fig:rhoplot},
the resource usage $u$ corresponding to Eq.~(\ref{eq:rhocore}) is
shown. We observe that under the small correlation degree, a
$20\%$ decrease of resource usage from
$u=\frac{\lambda}{\mu}(1+\delta)\approx 0.93$ (when $\Delta=0$) to
$u=\frac{\lambda}{\mu}=0.75$ (when $\Delta=\infty$) yields a
dramatic increase of the $99^{\textrm{th}}$ percentile of the
response times of roughly $230\%$. In turn, under the high
correlation degree, the same $20\%$ decrease of resource usage
from $u\approx 1.31$ (when $\Delta=0$) to  $u\approx 1.05$ (when
$\Delta=0.8$) yields an increase of the same response time
percentile of only roughly $37\%$. These numerical results, which
are clearly dependent on the model's assumptions and numerical
values, indicate nevertheless that a drastic reduction of the top
percentiles of response times at the expense of a small increase
of resource usage~\cite{Dean13,Jalaparti13} is due to a low
correlation of the service times. Conversely, if the service times
of the replicas are sufficiently correlated, increasing the
resource usage only yields a marginal gain in response time
reductions.

\begin{figure}[h]
 \begin{center}
  \includegraphics[width=0.48\textwidth]{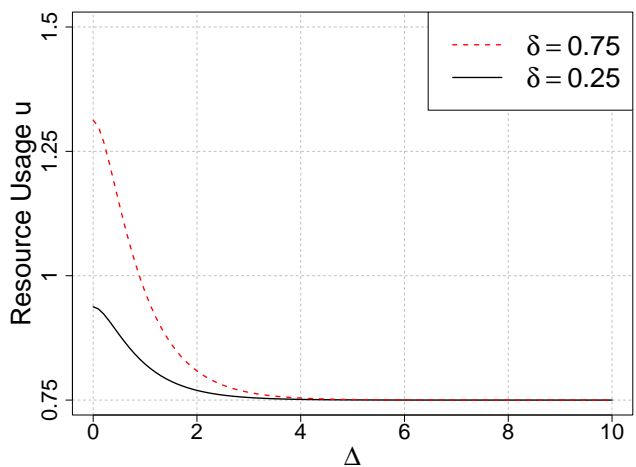}
  \caption{Resource usage $\rho$ from Eq.~(\ref{eq:rhocore}) ($\lambda=0.75$, $\mu=1$)}
  \label{fig:rhoplot}
 \end{center}
\end{figure}

\section{Conclusions}\label{sec:con}
In this paper we have developed an analytical framework to compute
stochastic bounds on the response time distribution in quite
general replicated queueing systems. Unlike existing models, ours
cover practical scenarios including correlated interarrivals,
general service time distributions, and not necessarily
independent service times for original tasks and their replicas.
By employing a powerful methodology based on martingale
transformations, we were able to derive numerically accurate
bounds by exploiting the specific correlation structures of the
underlying processes. Remarkably, we have shown both analytically
and through simulations that the choices of the underlying models
and assumptions play a fundamental role concerning the effects of
replication in parallel systems, thus motivating our general
framework. In terms of applications, we have developed a novel
task replication policy in fork-join systems which is similar to
the implementation of back-up tasks in MapReduce. For the
analytically convenient Poisson arrivals and i.i.d. exponential
service times model, our policy improves the performance of the
standard fork-join model by a fundamental logarithmic factor.

\bibliographystyle{abbrv}
\newcommand{\noopsort}[1]{}\providecommand{\noopsort}[1]{}

\end{document}